\documentclass[a4paper,final]{scrartcl}

\usepackage[utf8]{inputenc}
\usepackage[T1]{fontenc}

\usepackage{lmodern}
\usepackage{xcolor}
\usepackage{enumitem}
\usepackage{algorithm}
\usepackage{algorithmicx}
\usepackage[noend]{algpseudocode}
\usepackage{xspace}
\usepackage{amsmath}
\usepackage{amssymb}
\usepackage{amsthm}

\usepackage{tikz}
\usetikzlibrary{arrows}
\usetikzlibrary{matrix}

\newcommand{\set}[2]{\left\{#1\mathrel{\left|\vphantom{#1}\vphantom{#2}\right.}#2\right\}}
\newcommand{\os}[1]{\left\{\mathinner{#1}\right\}}
\newcommand{\abs}[1]{\left|\mathinner{#1}\right|}

\newcommand{\N}{\ensuremath{\mathbb{N}}}

\newcommand{\Oh}{\mathcal{O}}

\newcommand{\ms}{\hspace*{1pt}}

\newcommand{\trans}[1]{\mathbin{%
  \tikz[baseline=-0.5ex]\draw[->,>=stealth'] (0,0) -- %
  node[above,inner sep=0pt,outer sep=2pt]{\text{\small $#1$\,}} %
  (5mm, 0); %
  }}

\newcommand{\X}{\ensuremath\mathsf{X}}
\newcommand{\Y}{\ensuremath\mathsf{Y}}
\newcommand{\Xra}{$\X$\nobreakdash-ranker\xspace}
\newcommand{\Yra}{$\Y$\nobreakdash-ranker\xspace}
\newcommand{\Xras}{$\X$\nobreakdash-rankers\xspace}
\newcommand{\Yras}{$\Y$\nobreakdash-rankers\xspace}

\newcommand{\paletter}[4]{\hspace*{1mm}\overset{#2\,#3}{\underset{\color{gray!60} {\scriptscriptstyle #4}}{\vphantom{f}\vphantom{b}#1}\hspace*{1mm}}}
\newcommand{\aletter}[3]{\hspace*{1mm}\overset{#2\,#3}{\vphantom{b}#1}\hspace*{1mm}}
\newcommand{\maletter}[3]{\hspace*{1mm}\overset{\color{gray!60} #2\,#3}{\vphantom{b}#1}\hspace*{1mm}}

\newcommand{\ie}{\textit{i.e.}\xspace}

\newtheorem{remark}{Remark}
\newtheorem{example}{Example}
\newtheorem{theorem}{Theorem}
\newtheorem{lemma}[theorem]{Lemma}

\newtheorem{proposition}[theorem]{Proposition}

\title{Testing Simon's congruence}

\author{Lukas Fleischer\footnote{Supported by the German Research Foundation (DFG) under grant DI 435/5--2.}\\
  FMI, University of Stuttgart\\
  Universitätsstraße 38, 70569 Stuttgart, Germany\\
  \texttt{fleischer@fmi.uni-stuttgart.de}\\[1em]
  Manfred Kuf\-leitner\\
  Department of Computer Science, Loughborough University\\
  Epinal Way, Loughborough LE11 3TU, United Kingdom\\
  \texttt{m.kufleitner@lboro.ac.uk}}
\date{}

\begin{document}

\maketitle

\begin{abstract}
  Piecewise testable languages are a subclass of the regular languages. There are many equivalent ways of defining them; Simon's congruence $\sim_k$ is one of the most classical approaches. Two words are $\sim_k$-equivalent if they have the same set of (scattered) subwords of length at most $k$. A language $L$ is piecewise testable if there exists some $k$ such that $L$ is a union of $\sim_k$-classes. 

  For each equivalence class of $\sim_k$, one can define a canonical
  representative in shortlex normal form, that is, the minimal word with respect to the lexicographic order among the shortest words in $\sim_k$. We present an algorithm for computing the canonical representative of the $\sim_k$-class of a given word $w \in A^*$ of length $n$. The running time of our algorithm is in $\Oh(\abs{A} n)$ even if $k \leq n$ is part of the input. This is surprising since the number of possible subwords grows exponentially in $k$. The case $k>n$ is not interesting since then, the equivalence class of $w$ is a singleton. If the alphabet is fixed, the running time of our algorithm is linear in the size of the input word. Moreover, for fixed alphabet, we show that the computation of shortlex normal forms for $\sim_k$ is possible in deterministic logarithmic space.
  
  One of the consequences of our algorithm is that one can check with the same complexity whether two words are $\sim_k$-equivalent (with $k$ being part of the input).
\end{abstract}

\section{Introduction}

We write $u \prec v$ if the word $u$ is a (scattered) subword of $v$, that is, if there exist factorizations $u = u_1 \cdots u_n$ and $v = v_0 u_1 v_1 \cdots u_n v_n$. In the literature, subwords are sometimes called \emph{piecewise} subwords to distinguish them from factors. Higman showed that, over finite alphabets, the relation $\prec$ is a well-quasi-ordering~\cite{hig52}. This means that every language contains only finitely many minimal words with respect to the subword ordering. This led to the consideration of piecewise testable languages. A language $L$ is \emph{piecewise testable} if there exists a finite set of words $T$ such that $v \in L$ only depends on $\set{u \in T}{u \prec v}$; in other words, the occurrence and non-occurrence of subwords in $T$ determines membership in $L$. Equivalently, a language $L$ is piecewise testable if it is a finite Boolean combination of languages of the form $A^* a_1 A^* \cdots a_n A^*$ with $a_i \in A$ (using the above notation, the sequences $a_1 \cdots a_n$ in this combination give the words in $T$). The piecewise testable languages are a subclass of the regular languages and they play a prominent role in many different areas. For instance, they correspond to the languages definable in alternation-free first-order logic~\cite{tho82:short} which plays an important role in database queries. They also occur in learning
theory~\cite{Kontorovich08,RuizG96} and computational
linguistics~\cite{FuHT11,RogersHBEVWW10}.

In the early 1970s, Simon proved his famous theorem on piecewise testable languages: A language is piecewise testable if and only if its syntactic monoid is finite and $\mathcal{J}$-trivial~\cite{sim75:short}. An immediate consequence of Simon's Theorem is that it is decidable whether or not a given regular language $L$ is piecewise testable. Already in his PhD thesis~\cite{Simon72}, Simon considered the complexity of this problem when $L$ is given as a deterministic finite automaton~(DFA). His  algorithm can be implemented to have a running time of $\Oh(2^{\abs{A}} n^2)$ for an $n$-state DFA over the alphabet~$A$. This result was successively improved over the years~\cite{chohuynh91,KlimaPolak13,sternTCS85,Trahtman01} with the latest algorithm having a running time of $\Oh(\abs{A}^2 n)$; see~\cite{KlimaEtAl}. If the input is a DFA, then the problem is $\mathrm{NL}$-complete~\cite{chohuynh91}; and
if the input is a nondeterministic finite automaton, the problem is $\mathrm{PSPACE}$-complete~\cite{HolubMT14}. Restricting the length of the relevant subwords $T$ to some constant $k$ leads to the notion of $k$-piecewise testable languages. At first sight, it is surprising that, for every fixed $k \geq 4$, deciding whether a given DFA accepts a $k$-piecewise testable language is $\mathrm{coNP}$-complete~\cite{KlimaEtAl}; see also~\cite{MasopustT15}.

One of the main tools in the original proof of Simon's Theorem is the congruence $\sim_k$ for $k \in \N$. By definition, two words $u$ and $v$ satisfy $u \sim_k v$ if $u$ and $v$ have the same subwords of length at most $k$. Naturally, the relation~$\sim_k$ is nowadays known as \emph{Simon's congruence}. It is easy to see that a language $L$ is piecewise testable if and only if there exists $k$ such that~$L$ is a union of $\sim_k$-classes. Understanding the combinatorial properties of $\sim_k$ is one of the main tools in the study of piecewise testable languages. For example, in the proof of his theorem, Simon already used that $(uv)^k \sim_k (uv)^k u$ for all words $u,v$. Upper and lower bounds on the index of $\sim_k$ were given by K{\'a}tai-Urb{\'a}n et al.~\cite{KataiPPPS2012} and Karandikar et al.~\cite{KaranfikarKS2015}.

There are two natural approaches for testing whether or not $u \sim_k v$ holds. The first approach constructs a DFA $\mathcal{A}_{k,u}$ for the language $\set{w \prec u}{k \geq \abs{w}}$ of the subwords of $u$ of length at most $k$ and a similar DFA $\mathcal{A}_{k,v}$ for $v$. Then $u \sim_k v$ if and only if $\mathcal{A}_{k,u}$ and $\mathcal{A}_{k,v}$ accept the same language. This can be tested with Hopcroft's algorithm in time almost linear in the size of the automata~\cite{hoka71}. Here, \emph{almost linear} in $n$ means $\Oh(n \cdot a(n))$ where~$a(n)$ is the inverse Ackermann function. It is possible to construct the automata such that $\mathcal{A}_{k,u}$ has at most $k \abs{u} + 2$ states, see the remark at the end of Section~\ref{sec:prelim} below. Hence, the resulting test is almost linear in $\abs{A} \ms k \ms \abs{uv}$ if the alphabet is $A$.

The second approach to testing $u \sim_k v$ is the computation of normal forms. A normal form is a unique representative of a $\sim_k$-class. In particular, we have $u \sim_k v$ if and only if $u$ and $v$ have the same normal form. By computing the normal forms for both words and then checking whether they are identical, the complexity of this test of $u \sim_k v$ is the same as the computation of the normal forms. We should mention that the computation of normal forms is also interesting in its own right since it can provide some insight into the combinatorial properties of $\sim_k$.
Normal forms for $k=2$ and $k=3$ were considered by K{\'a}tai-Urb{\'a}n et al.~\cite{KataiPPPS2012} and normal forms for $k=4$ were given by Pach~\cite{Pach2015}. An algorithm for computing normal forms for arbitrary $k$ was found
only recently by Pach~\cite{Pach2017}. Its running time is $\Oh(\abs{A}^k  (n + \abs{A}))$ for inputs of length $n$ over the alphabet $A$, that is, polynomial for fixed $k$ and exponential otherwise.

We significantly improve this result by providing an algorithm with a running
time in $\Oh(\abs{A} n)$ even if $k$ is part of the input.
For a fixed alphabet, the running time is linear which is optimal.
Moreover, the algorithm can easily be adapted to run in deterministic
logarithmic space, thereby addressing an open problem from~\cite{KataiPPPS2012}. As a consequence we can check with the same running time (or the same complexity) whether two given words are~$\sim_k$-equivalent even if $k$ is part of the input, thereby considerably improving on the above automaton approach.

Our algorithm actually does not compute just some normal form but the \emph{shortlex normal form} of the input word $u$, \ie,~the shortest, and among all shortest the lexicographically smallest, word $v$ such that $u \sim_k v$.
Our main tools are so-called \emph{rankers}~\cite{stv01dlt:short,wi09lmcs}.
For each position $i$ in the input word, the algorithm computes the lengths of
the shortest \Xras and \Yras reaching $i$. One can then derive the
shortlex normal form by deleting and sorting certain letters based on these
attributes. A more detailed outline of the paper is given in Section~\ref{sec:outline}.

\section{Preliminaries}\label{sec:prelim}

Let $A$ be a finite alphabet. The elements in $A$ are called \emph{letters} and
a sequence of letters $u = a_1 \cdots a_\ell$ is a \emph{word}. The number
$\ell$ is the \emph{length} of the word. It is denoted by $\abs{u}$.
The set of all words over the alphabet $A$ is $A^*$.
Throughout this paper,  $a$, $b$ and $c$ are used to denote
letters.
For a word $a_1 \cdots a_\ell$, the numbers $\os{1, \dots, \ell}$ are called
\emph{positions} of the word, and $i$ is a \emph{$c$-position} if
$a_i = c$. The letter $a_i$ is the \emph{label} of position $i$. Two positions $i$ and $j$ with $i<j$ are \emph{consecutive} $c$-positions if $a_i = a_j = c$ and $a_\ell \neq c$ for all $\ell \in \os{i+1,\ldots,j-1}$.

A word $a_1 \cdots a_\ell$ is a \emph{subword}
of a word $v \in A^*$ if $v$ can be written as $v = v_0 a_1 \cdots v_{\ell-1}
a_\ell v_\ell$ for words $v_i \in A^*$. We write $u \prec v$ if $u$ is a subword
of $v$.
A \emph{congruence} on $A^*$ is an equivalence relation $\sim$ such that $u
\sim v$ implies $puq \sim pvq$ for all $u, v, p, q \in A^*$.
For a fixed number $k \in \N$, \emph{Simon's congruence} $\sim_k$ on $A^*$ is defined by $u
\sim_k v$ if and only if $u$ and $v$ contain the same subwords of length at
most $k$.

We assume that the letters of the alphabet are totally ordered.
A word $u$ is \emph{lexicographically smaller than $v$} if, for some common $p
\in A^*$, there exists a prefix $pa$ of $u$ and a prefix~$pb$ of~$v$ such that
$a < b$. (We apply the lexicographic order only for words of the same length;
in particular, we do not care about the case when $u$ is a proper prefix of
$v$.)
Given a congruence~$\sim$ on $A^*$, we define the \emph{shortlex normal form}
of a word $u$ to be the shortest word $v$ such that $u \sim v$ and such that no
other word $w \in A^*$ with $w \sim v$ and $\abs{w} = \abs{v}$ is
lexicographically smaller than $v$. In other words, we first pick the shortest
words in the~$\sim$-class of $u$ and among those, we choose the
lexicographically smallest one.

Our main tools are so-called \emph{rankers}~\cite{stv01dlt:short,wi09lmcs}.
An \emph{\Xra} is a nonempty word over the alphabet $\set{\X_a}{a \in A}$ and
a \emph{\Yra} is a nonempty word over $\set{\Y_a}{a \in A}$. The length of a
ranker is its length as a word.
The modality $\X_a$ means ne$\X$t-$a$ and is interpreted as an instruction of
the form ``go to the next $a$-position''; similarly, $\Y_a$ is a shorthand for~$\Y$esterday-$a$ and means ``go to the previous $a$-position''.
More formally, we let $\X_a(u) = i$ if~$i$ is the smallest
$a$\nobreakdash-position of $u$, and we let $r \X_a(u) = i$ for a ranker $r$ if $i$ is the
smallest $a$\nobreakdash-position greater than $r(u)$.
Symmetrically, we let $\Y_a(u) = i$ if $i$ is the greatest $a$\nobreakdash-position of $u$ and
we let $r \Y_a(u) = i$ if $i$ is the greatest $a$\nobreakdash-position smaller
than $r(u)$. In particular, rankers are processed from left to right.
Note that the position $r(u)$ for a ranker $r$ and a word $u$ can be undefined.
A word $b_1 \cdots b_\ell$ \emph{defines} an \Xra $\X_{b_1} \cdots \X_{b_\ell}$
and a \Yra $\Y_{b_\ell} \cdots \Y_{b_1}$. We have $u \prec v$ if and only if
$r(v)$ is defined for the \Xra (resp.\ \Yra) $r$ defined by $u$. 
Similarly, if $r$ is the \Xra defined by $u$ and $s$ is the \Yra defined by
$v$, then $uv \prec w$ if and only if $r(w) < s(w)$. The correspondence between
\Xras and subwords leads to the following automaton construction.

\begin{remark}\label{rem:smallautomata}
Let $u$ be a word of length $n$. We construct a DFA $\mathcal{A}_{k,u}$ for the language $\set{w \prec u}{\abs{w} \leq k}$. The set of states is $\os{(0,0)} \cup \os{1,\ldots,k} \times \os{1,\ldots,n}$ plus some sink state which collects all missing transitions. The initial state is $(0,0)$ and all states except for the sink state are final. We have a transition $(\ell,i) \trans{a} (\ell+1,j)$ if $\ell < k$ and $j$ is the smallest $a$-position greater than $i$. The idea is that the first component counts the number of instructions and the second component gives the current position.
\qed
\end{remark}

\section{Attributes and outline of the paper}\label{sec:outline}

To every position $i \in \os{1,\ldots,n}$ of a word $a_1 \cdots a_n \in A^n$, we assign an \emph{attribute} $(x_i,y_i)$ where $x_i$ is the length of a shortest \Xra reaching $i$ and $y_i$ is the length of a shortest \Yra reaching $i$. We call $x_i$ the \emph{$x$-coordinate} and $y_i$ the \emph{$y$-coordinate} of position $i$.

\begin{example}\label{ex:second}
We will use the word $u =bacbaabada$ as a running example throughout this paper. 
The attributes of the positions in $u$ are as follows:
\begin{equation*}
  \aletter{b}{1}{2}
  \aletter{a}{1}{2}
  \aletter{c}{1}{1}
  \aletter{b}{2}{2}
  \aletter{a}{2}{3}
  \aletter{a}{3}{2}
  \aletter{b}{3}{1}
  \aletter{a}{4}{2}
  \aletter{d}{1}{1}
  \aletter{a}{2}{1}
\end{equation*}
The letter $a$ at position $5$ can be reached by the \Yra $\Y_b \Y_a \Y_a$ and
the $a$ at position~$6$ can be reached by the \Xra $\X_c \X_a \X_a$.
Both rankers visit both positions $5$ and $6$. No \Xra visiting position $6$
can avoid position $5$ and no \Yra visiting position $5$ can avoid position
$6$. Deleting either position $5$ or position $6$ reduces the attributes of the
other position to $(2, 2)$.
\qed
\end{example}

We propose a two-phase algorithm for computing the shortlex normal of a word $u$ within its $\sim_k$-class. The first phase is to reduce the word by deleting letters resulting in a word of minimal length within the $\sim_k$-class of $u$. The second phase sorts blocks of letters to get the minimal word with respect to the lexicographic ordering.
Both phases depend on the attributes. The computation of the attributes and the first phase are combined as follows.

\begin{description}
\item[Phase 1a:] Compute all $x$-coordinates from left to right.
\item[Phase 1b:] Compute all $y$-coordinates from right to left while dynamically deleting a position whenever the sum of its coordinates would be bigger than $k+1$.
\item[Phase 2:] Commute consecutive letters $b$ and $a$ (with $b>a$) whenever they have the same attributes and the sum of the $x$- and the $y$-coordinate equals $k+1$.
\end{description}
As we will show, a crucial property of Phase~1b is that the dynamic process does not mess up the $x$-coordinates of the remaining positions that were previously computed in Phase~1a.

The \textsf{\textbf{outline}} of the paper is as follows. In Section~\ref{sec:lengthred}, we prove that successively deleting all letters where the sum of the attributes is bigger than $k+1$ eventually yields a length-minimal word within the $\sim_k$-class of the input. This statement has two parts. The easier part is to show that we can delete such a position without changing the $\sim_k$-class. The more difficult part is to show that if no such deletions are possible, the word is length-minimal within its $\sim_k$-class. In particular, no other types of deletions are required. Also
note that deleting letters can change the attributes of the remaining letters.

Section~\ref{sec:comm} has two components. First, we show that commuting consecutive letters does not change the $\sim_k$-class if (a) the two letters have the same attribute and (b) the sum of the $x$- and the $y$-coordinate equals $k+1$. Moreover, such a commutation does not change any attributes. Then, we prove that no other types of commutation are possible within the $\sim_k$-class. This is quite technical to formalize since, a priori, we could temporarily leave the $\sim_k$-class only to re-enter it again with an even smaller word.

Finally, in Section~\ref{sec:alg}, we present an easy and efficient algorithm for computing shortlex normal forms for $\sim_k$. First, we show how to efficiently compute the attributes. Then we combine this computation with a single-pass deletion procedure; in particular, we do not have to successively re-compute the attributes after every single deletion. Finally, an easy observation shows that we only have to sort disjoint factors where the length of each factor is bounded by the size of the alphabet. Altogether, this yields an $\Oh(\abs{A} n)$ algorithm for computing the shortlex normal form of an input word of length $n$ over the alphabet $A$. Surprisingly, this bound also holds if $k$ is part of the input.

\section{Length reduction}\label{sec:lengthred}

In order to reduce words to shortlex normal form, we want to identify positions
in the word which can be deleted without changing its $\sim_k$-class.
The following proposition gives a sufficient condition for such deletions.

\begin{proposition}\label{prp:delLetter}
Consider a word $uav$ with $a \in A$ and $\abs{ua} = i$. If the attribute $(x_i,y_i)$ at position $i$ satisfies $x_i + y_i > k+1$, then $uav \sim_k uv$.
\end{proposition}

\begin{proof}
Let $w \prec uav$ with $\abs{w} \leq k$. Assume that $w \not\prec uv$. Let $w = paq$ such that $p \prec u$ and $q \prec v$. Note that $pa \not\prec u$ and $aq \not\prec v$. If $\abs{p} \geq x_i -1$ and $\abs{q} \geq y_i - 1$, then 
\begin{equation*}
  k \geq \abs{w} = \abs{p} + 1 + \abs{q} \geq (x_i - 1) + 1 + (y_i - 1) = x_i + y_i - 1 > k,
\end{equation*}
a contradiction. Therefore, we have either $\abs{p} < x_i - 1$ or $\abs{q} < y_i - 1$. By left-right symmetry, it suffices to consider the case $\abs{p} < x_i - 1$. The word $pa$ defines an \Xra of length less than $x_i$ which reaches position $i$. This is not possible by definition of $x_i$. Hence, $w \prec uv$. Conversely, if $w \prec uv$ for a word $w$, then obviously we have $w \prec uav$. This shows  $uav \sim_k uv$.
\end{proof}

\begin{example}\label{ex:double-del}
Let $u = bacbaabada$ as in Example~\ref{ex:second} and let $k = 3$.
Note that the attributes $(x_i, y_i)$ at positions $i \in \os{5, 6}$ satisfy
the condition $x_i + y_i > k + 1$. By Proposition~\ref{prp:delLetter}, deleting
any of these positions yields a $\sim_k$-equivalent word. However, deleting
both positions yields the word $bacbbada \not\sim_k u$ since $cab \prec u$ and
$cab \not\prec bacbbada$.
\qed
\end{example}

\noindent
Consider a position $i$ with label $c$ and attribute $(x_i,y_i)$ in a word~$u$. Let
\begin{equation*}
  R^u_i =  \set{r}{r \text{ is an \Xra with } r(u) = i \text{ and } \abs{r} = x_i}.
\end{equation*}
We have $R^u_i \neq \emptyset$ by definition of $x_i$. We define a \emph{canonical \Xra} $r^u_i \in R^u_i$ by minimizing the reached positions, and the minimization procedure goes from right to left: Let $S_{x_i} = R^u_i$ and, inductively, we define $S_j$ as a nonempty subset of $S_{j+1}$ as follows. Let $p_j$ be the minimal position in $u$ visited by the  prefixes $s$ of length $j$ of the rankers in $S_{j+1}$; then $S_j$ contains all rankers in $S_{j+1}$ such that their prefixes of length $j$ visit the position $p_j$. Since the minimal positions (and their labels) in this process are unique, we end up with $\abs{S_1} = 1$. Now, the ranker $r^u_i$ is given by $S_1 = \os{r^u_i}$. By abuse of notation, we will continue to use the symbol $r$ for arbitrary rankers while $r^u_i$ denotes canonical rankers. The following example shows that minimizing from right to left (and not the other way round) is crucial.

\begin{example}
Let $u = abcab cdaef ccabc$.
The attributes of the letters are as follows:
\begin{equation*}
  \paletter{a}{1}{3}{1}
  \paletter{b}{1}{3}{2}
  \paletter{c}{1}{3}{3}
  \paletter{a}{2}{2}{4}
  \paletter{b}{2}{2}{5}
  \paletter{c}{2}{2}{6}
  \paletter{d}{1}{1}{7}
  \paletter{a}{2}{2}{8}
  \paletter{e}{1}{1}{9}
  \paletter{f}{1}{1}{10}
  \paletter{c}{2}{3}{11}
  \paletter{c}{3}{2}{12}
  \paletter{a}{2}{1}{13}
  \paletter{b}{2}{1}{14}
  \paletter{c}{3}{1}{15}
\end{equation*}
The last $c$ is at position $15$ and its attribute is $(3,1)$. It is easy to
verify that $\X_e \X_a \X_c$ is an \Xra of length $3$ visiting position $15$
and that there is no \Xra of length $2$ reaching this position. The unique \Yra
of length $1$ reaching position $15$ is $\Y_c$.
We have
\begin{equation*}
  R^u_{15} = \os{\X_d \X_b \X_c, \X_e \X_a \X_c, \X_e \X_b \X_c, \X_f  \X_a \X_c, \X_f \X_b \X_c}.
\end{equation*}
Using the above notation, it is easy to see that $S_3 = R^u_{15}$, $S_2 = \os{\X_e \X_a \X_c, \X_f  \X_a \X_c}$, and $S_1 = \os{\X_e \X_a \X_c}$. All prefixes of length $2$ of rankers in $S_2$ reach position $p_2 = 13$; the prefix of length $1$ of the ranker in $S_1$ reaches position $p_1 = 9$. The ranker visiting positions $9$, $13$ and $15$ (and no other positions) is $r^u_{15} = \X_e \X_a \X_c$, the unique ranker in $S_1$.

Also note that the minimal positions $m_j$ visited by prefixes of length $j$ of the rankers in~$R^u_{15}$ are $m_1 = 7$, $m_2 = 13$, and $m_3 = 15$; but there is no single ranker of length $3$ visiting positions $7$, $13$, and $15$.
\qed
\end{example}

While $r^u_i$ is defined in some right-to-left manner, it still has an important left-to-right property when positions of the same label are considered. 

\begin{lemma}\label{lem:consecutive}
Let $i<j$ be two consecutive $c$-positions in a word $u$ with attributes $(x_i,y_i)$ and $(x_j,y_j)$, respectively. If $x_j > x_i$, then $r^u_j = r^u_i \X_c$. 
\end{lemma}

\begin{proof}
Since no position $\ell$ with $i < \ell < j$ is labelled by $c$, we have $r^u_i \X_c(u) = j$. In particular, $x_j = x_i + 1$ and $R^u_i \X_c \subseteq R^u_j$. Let $r^u_j = r \X_c$. We have $r(u) \geq r^u_i(u)$, since otherwise $r \X_c(u) \leq i < j$, a contradiction. The minimization in the definition of $r^u_j$ now yields $r(u) = r^u_i(u)$. The remaining minimization steps in the definition of $r^u_i$ and $r^u_j$ consider the same rankers and thus the same positions. Hence, $r = r^u_i$.
\end{proof}

We now want to prove that the condition introduced in Proposition~\ref{prp:delLetter}
always results in a shortest word within the corresponding $\sim_k$-class.
To this end, we first need the following technical lemma and then prove the
main theorem of this section.

\begin{lemma}\label{lem:orderpreserve}
Let $u = a_1 \cdots a_n$ be a word and let $i<j$ be positions with $a_i = a_j$ and with $a_\ell \neq a_i$ for all $\ell \in \os{i+1,\ldots,j-1}$, \ie, $i$ and $j$ are consecutive $a_i$\nobreakdash-positions. Moreover, let the parameters $(x_i,y_i)$ and $(x_j,y_j)$ satisfy $x_i + y_i \leq k+1$ and $x_j \leq k$, respectively.
 For every word $v$ with $u \sim_k v$, we have $r^u_i(v) < r^u_j(v)$.
\end{lemma}

\begin{proof}
We have $x_i \leq k$ and $x_j \leq k$. Therefore, 
both $r^u_i(v)$ and $r^u_j(v)$ are defined because this only depends on subwords of length at most $k$ which are identical for $u$ and $v$. Let $c = a_i$.
Since $j = r^u_i \ms \X_c (u)$, we have $x_j \leq x_i + 1$. If $x_j = x_i + 1$, then $r^u_j = r^u_i \ms \X_c$ by Lemma~\ref{lem:consecutive} and hence $r^u_i(v) < r^u_j(v)$. Therefore, we can assume $x_j \leq x_i$. Suppose that $r^u_i(v) \geq r^u_j(v)$. Let $q \ms \Y_c$ be a \Yra with $q \ms \Y_c(u) = i$ and $\abs{q \ms \Y_c} = y_i$. Let $w_i$ be the word corresponding to $r^u_i$, let $w_j$ be the word corresponding to $r^u_j$, and let $z$ be the word corresponding to $q$. We have:
\begin{align*}
& w_i z \prec u \qquad && \text{since $r^u_i(u) = i < q(u)$}
\\ \Rightarrow \ & w_i z \prec v && \text{since $\abs{w_i z} = x_i + y_i - 1 \leq k$ and $u \sim_k v$}
\\ \Rightarrow \ & w_j z \prec v && \text{since $r^u_i(v) \geq r^u_j(v)$}
\\ \Rightarrow \ & w_j z \prec u && \text{since $\abs{w_j z} = x_j + y_i - 1 \leq x_i + y_i -1 \leq k$}
\\ \Rightarrow \ & q(u) > j
\\ \Rightarrow \ & q\ms \Y_c(u) \geq j > i.
\end{align*}
This contradicts $q \Y_c(u) = i$. Therefore, we have $r^u_i(v) < r^u_j(v)$.
\end{proof}

\begin{theorem}\label{thm:del}
  If $u$ is a word such that the attribute $(x_i,y_i)$ of every position $i$ satisfies $x_i + y_i \leq k+1$, then $u$ has minimal length within its $\sim_k$-class.
  \label{thm:minlen}
\end{theorem}

\begin{proof}
Let $v$ be a shortest word satisfying $u \sim_k v$. Let $\rho$ map the position $j$ of $u$ to the position $r^u_j(v)$ of $v$. Consider some letter $c$ occurring in~$u$. Then, by Lemma~\ref{lem:orderpreserve}, the function $\rho$ maps the $i$-th occurrence of the letter $c$ in $u$ to the $i$-th occurrence of the letter $c$ in $v$. In particular, the word $v$ has at least as many occurrences of $c$ as $u$. This holds for all letters $c$ in $u$, hence, $\abs{u} \leq \abs{v}$.
\end{proof}

\begin{example}\label{ex:length-minimal}
Consider $u = bacbaabada$ from Example~\ref{ex:second} and let $k = 3$. As
explained in Example~\ref{ex:double-del}, we must not delete both position $5$
and position $6$. However, we can delete positions $5$ and $8$ to obtain a
$\sim_k$-equivalent word with the following attributes:
\begin{equation*}
  \aletter{b}{1}{2}
  \aletter{a}{1}{2}
  \aletter{c}{1}{1}
  \aletter{b}{2}{2}
  \aletter{a}{2}{2}
  \aletter{b}{3}{1}
  \aletter{d}{1}{1}
  \aletter{a}{2}{1}
\end{equation*}
By Theorem~\ref{thm:minlen}, there is no shorter word in the same
$\sim_k$-class.
\qed
\end{example}

\section{Commutation}\label{sec:comm}

In the previous section, we described how to successively delete letters of a
word in order to obtain a length-minimal $\sim_k$-equivalent word.
It remains to show how to further transform a word of minimal length into shortlex normal form.
In the first two lemmas, we give a sufficient condition which allows us to
commute letters $b$ and $a$ while preserving the $\sim_k$-class.

\begin{lemma}\label{lem:commparam}
  Consider two words $ubav$ and $uabv$ with $a,b \in A$. Let $(x_\ell,y_\ell)$ denote the attribute of position $\ell$ in $ubav$, and let $(x'_\ell,y'_\ell)$ denote the attribute of position $\ell$ in $uabv$. Suppose that $\abs{ub} = i$ and that the attributes $(x_i,y_i)$ and $(x_{i+1},y_{i+1})$ satisfy $x_i = x_{i+1}$. Then all positions $\ell$ satisfy $x'_\ell = x_\ell$.
\end{lemma}

\begin{proof}
We can assume that $a\neq b$.
It suffices to show that no ranker in $R^{ubav}_{i+1}$ visits position $i$ in
$ubav$ and no ranker in $R^{uabv}_{i+1}$ visits position $i$ in $uabv$.
This implies that for all $\ell \in \os{1, \dots, n}$, no ranker in
$R^{ubav}_\ell$ or in $R^{uabv}_\ell$ visits both $i$ and $i+1$ in the
corresponding words and thus, we have $R^{ubav}_{i} = R^{uabv}_{i+1}$ and
$R^{ubav}_{i+1} = R^{uabv}_i$ as well as $R^{ubav}_\ell = R^{uabv}_\ell$ for
$\ell \not\in \os{i, i+1}$. Note that all rankers in $R^{ubav}_i$ and in
$R^{ubav}_{i+1}$ have length $x_i = x_{i+1}$.

Suppose, for the sake of contradiction, that a ranker $r \in R^{ubav}_{i+1}$
visits position $i$ in $ubav$.
Then, we can write $r = s \X_b \X_a$ with $s\X_b(ubav) = i$. Note that
$\abs{s\X_b} \ge x_i$ by the definition of $x_i$.
Since $x_{i+1} = x_i \le \abs{s\X_b}$, there exists a ranker of length at most
$\abs{s\X_b} < \abs{r}$ reaching position $i+1$ in $ubav$, contradicting the
choice of $r$.

Suppose that a ranker $r \in R^{uabv}_{i+1}$ visits position $i$ in $uabv$.
Let $r = s \X_a \X_b$ with $s\X_a(uabv) = i$. Note that $s\X_a(ubav) =
i+1$ and, since $x_{i+1} = x_i$, there exists a ranker $\hat s$ of length at
most $\abs{s\X_a}$ such that $\hat s(ubav) = i$.
Now, $\hat s$ is a ranker of length $\abs{\hat s} \le \abs{s\X_a} < \abs{r}$
with $\hat s(uabv) = i+1$, a contradiction to $r \in R^{uabv}_{i+1}$.
\end{proof}

\begin{proposition}\label{prp:commneighbors}
  Let $ubav$ be a word with $\abs{ub} = i$ and attributes $(x_i,y_i)=(x_{i+1},y_{i+1})$ satisfying $x_i + y_i = k+1$.
  Then $ubav \sim_k uabv$.
\end{proposition}

\begin{proof}
Suppose that there exists a word $w$ with $\abs{w} \leq k$ such that $w \prec ubav$ but $w \not\prec uav$ and $w \not\prec ubv$. Then we can write $w = w_1 ba w_2$ such that $w_1 b \prec ub$, $w_1 b \not\prec u$, $aw_2 \prec av$, and $aw_2 \not\prec v$. Thus, the word $w_1 b$ defines an \Xra $r$ with $r(ubav) = i$ and, similarly, $a w_2$ defines a \Yra $s$ with $s(ubav) = i+1$. We see that $\abs{r} + \abs{s} = \abs{w} \leq k$, but this contradicts $\abs{r} + \abs{s} \geq x_i + y_{i+1} = k+1$. Therefore, every subword of $ubav$ of length at most~$k$ is also a subword of $uabv$.

By Lemma~\ref{lem:commparam} and its left-right dual, the attributes of the positions $i$ and $i+1$ in $uabv$ are both identical to $(x_i,y_i)$. Therefore, the same reasoning as above shows that every subword of $uabv$ of length at most $k$ is also a subword of $ubav$.  This shows $ubav \sim_k uabv$.
\end{proof}

\begin{example}\label{ex:commutation}
Let us reconsider the length-minimal word $u = bacbabda$ from
Example~\ref{ex:length-minimal} and let again $k = 3$. The attributes are as
follows:
\begin{equation*}
  \aletter{b}{1}{2}
  \aletter{a}{1}{2}
  \aletter{c}{1}{1}
  \aletter{b}{2}{2}
  \aletter{a}{2}{2}
  \aletter{b}{3}{1}
  \aletter{d}{1}{1}
  \aletter{a}{2}{1}
\end{equation*}
The attributes $(x_4, y_4)$ and $(x_5, y_5)$ at positions $4$ and $5$ satisfy
$x_4 = x_5$, $y_4 = y_5$ and $x_4 + y_4 = k+1$.
By Proposition~\ref{prp:commneighbors}, we obtain $bacabbda \sim_k u$.
Note that the attributes $(x_1, y_1)$ and $(x_2, y_2)$ at the first two
positions satisfy $x_1 = x_2$, $y_1 = y_2$ but $x_1 + x_2 < k+1$. And, in fact,
$abcbabda \not\sim_k u$ since $abc \prec abcbabda$ but $abc \not\prec u$.
\qed
\end{example}

It remains to show that repeated application of the commutation rule described
in Proposition~\ref{prp:commneighbors} actually suffices to obtain the
lexicographically smallest representative of a $\sim_k$-class.
The next lemma shows, using canonical rankers, that indeed all length-minimal
representatives of a $\sim_k$-class can be transformed into one another using
this commutation rule.

\begin{lemma}\label{lem:noncomm}
Let $u \sim_k v$ such that both words $u$ and $v$ have minimal length in their $\sim_k$-class. Let $(x_\ell,y_\ell)$ denote the attribute of position $\ell$ of $u$. Consider two positions $i<j$ of $u$. If either $(x_i,y_i) \neq (x_j,y_j)$ or $x_i + y_i < k+1$, then $r^u_i(v) < r^u_j(v)$.
\end{lemma}

\begin{proof}
If $i$ and $j$ have the same label, then the claim follows from Lemma~\ref{lem:orderpreserve}. In the remainder of this proof, let their labels be different. In particular, we cannot have $r^u_i(v) = r^u_j(v)$. Suppose $(x_i,y_i) \neq (x_j,y_j)$ or $x_i + y_i < k+1$. If $x_i + y_j \geq k+1$ and $x_j+y_i \geq k+1$, then, by minimality, $x_i + y_i = k+1$ and $x_j + y_j = k+1$. This yields $x_i + y_j = k+1$ and $x_j+y_i = k+1$. Thus, $x_i = x_i + x_j + y_j - k-1 = x_j$ and, similarly, $y_i = y_i + x_j + y_j - k-1 = y_j$; this shows $(x_i,y_i) = (x_j,y_j)$, a contradiction. Therefore, we have either $x_i + y_j \leq k$ or $x_j+y_i \leq k$.

Let $p_i$ and $p_j$ be the words defining the rankers $r^u_i$ and $r^u_j$, respectively. Symmetrically to the definition of the canonical \Xra, we could also define canonical \Yras $s^u_i$ and~$s^u_j$ such that $s^u_i(u) = i$, $\abs{s^u_i} = y_i$, $s^u_j(u) = j$, and $\abs{s^u_j} = y_j$. If the label $c$ of $u$ at position $i$ is the $\ell$-th occurrence of the letter $c$ in $u$, then, by Lemma~\ref{lem:orderpreserve}, both $r^u_i$ and $s^u_i$ end up at the position with the $\ell$-occurrence of the letter $c$ in $v$. This shows $r^u_i(v) = s^u_i(v)$. Similarly, we see that $r^u_j(v) = s^u_j(v)$. Let $q_i$ and $q_j$ be the words defining the rankers $s^u_i$ and $s^u_j$, respectively.

First, let $x_i + y_j \leq k$. Then $p_i q_j \prec u$ yields $p_i q_j \prec v$ since $u \sim_k v$ and $\abs{p_i q_j} = x_i + y_j \leq k$. This shows $r^u_i(v) < s^u_j(v) = r^u_j(v)$, as desired. 
Let now $x_j + y_i \leq k$ and assume $r^u_i(v) > r^u_j(v)$. Then $p_j q_i \prec v$ yields $p_j q_i \prec u$ and, thus, $j = r^u_j(u) < s^u_i(u) = i$. This is a contradiction; hence, $r^u_i(v) < r^u_j(v)$.
\end{proof}

Using the previous lemma, we can finally show that iterating the commutation procedure from Lemma~\ref{lem:commparam} and Proposition~\ref{prp:commneighbors} yields the desired shortlex normal form.

\begin{theorem}\label{thm:comm}
  Let $u=a_1 \cdots a_n$ with $a_i \in A$ be a length-minimal word within its $\sim_k$-class. Suppose that the attributes $(x_i,y_i)$ for all positions $i<n$ satisfy the following implication:
  \begin{equation}
    \text{If } (x_i,y_i) = (x_{i+1},y_{i+1}) \text{ and } x_i + y_i = k+1, \text{ then } a_i \leq a_{i+1}.
    \label{eq:comm}
  \end{equation}
  Then $u$ is the shortlex normal of its $\sim_k$-class.
\end{theorem}

\begin{proof}
  Let $v$ be the shortlex normal form of the $\sim_k$-class of $u$. We want to show that $u=v$. Let $\rho$ map position $i$ of $u$ to position $r^u_i(v)$ of $v$. As we have seen in the proof of Theorem~\ref{thm:del}, the function $\rho$ is bijective. It remains to show that $\rho$ is order-preserving. By contradiction, assume that there are positions $i$ and $j$ of $u$ with $i<j$ such that $\rho(i) > \rho(j)$; let $i$ be minimal with this property and let $i = \rho(j)$, \ie, we choose $j$ to be the preimage of position $i$ in $v$.
  We already know that $\rho(i) < \rho(j)$ in all of the following cases:
  \begin{itemize}
  \item $a_i = a_j$ \,(by Lemma~\ref{lem:orderpreserve}),
  \item $(x_i,y_i) \neq (x_j,y_j)$ \,(by Lemma~\ref{lem:noncomm}),
  \item $x_i + y_i < k+1$ \,(again by Lemma~\ref{lem:noncomm}).
  \end{itemize}
  Therefore, the only remaining case is $a_i \neq a_j$, $(x_i,y_i) = (x_j,y_j)$ and $x_i + y_i = k+1$. First, suppose that $(x_i,y_i) = (x_\ell,y_\ell)$ for all $\ell \in \os{i,\ldots,j}$. Then, by the implication in Equation~\eqref{eq:comm}, we have $a_i \leq \cdots \leq a_j$. Since $a_i \neq a_j$, we have $a_i < a_j$. Now, $u$ has the prefix $a_1 \cdots a_i$ and $v$ has the prefix $a_1 \cdots a_{i-1} a_j$. In particular, $u$ is lexicographically smaller than $v$; this is a contradiction. Next, suppose that there exists a position $\ell \in \os{i,\ldots,j}$ with $(x_i,y_i) \neq (x_\ell,y_\ell)$. Note that $i < \ell < j$. By Lemma~\ref{lem:noncomm}, we have $\rho(i) < \rho(\ell)$ and $\rho(\ell) < \rho(j)$. In particular, we have $\rho(i) < \rho(j)$ in contradiction to our assumption. Altogether, this shows that the situation $i<j$ and $\rho(i) > \rho(j)$ is not possible, \ie, $\rho$ is order-preserving. Hence, $u=v$ as desired.
\end{proof}

We summarize our knowledge on shortest elements of a $\sim_k$-class as follows.
A word $u$ has minimal length within its $\sim_k$-class if and only if all
attributes $(x_i,y_i)$ satisfy $x_i + y_i \leq k+1$. The canonical rankers
define a bijective mapping between any two shortest words $u$ and $v$ of a
common $\sim_k$-class. This map preserves the labels and the attributes. It is
almost order preserving, with the sole exception that $i<j$ could lead to
$r^u_i(v) > r^u_j(v)$ whenever the attributes in $u$ satisfy both $x_i + y_i =
k+1$ and $(x_i,y_i) = (x_\ell,y_\ell)$ for all $\ell \in \os{i,\ldots,j}$.

\section{Computing shortlex normal forms}\label{sec:alg}

The results from the previous sections immediately lead to the following algorithm for computing shortlex normal forms. First, we successively delete single letters of the input word until the length is minimal. Let $a_1 \cdots a_n$ be the resulting word. In the second step, we lexicographically sort maximal factors $a_i \cdots a_j$ with attributes $(x_i,y_i) = \cdots = (x_j,y_j)$ and $x_i + y_i = k+1$.
We now improve the first step of this algorithm.

\begin{algorithm}[h]
\caption{Computing the $x$-coordinates of $a_1 \cdots a_n$}\label{alg:xcoord}
\begin{algorithmic}[1]
  \ForAll {$a \in A$} $n_a \gets 1$
  \EndFor
  \For {$i \gets 1,\ldots,n$} 
  \State suppose $a_i = c$
  \State $x_i \gets n_c$
  \State $n_c \gets  n_c + 1$
  \ForAll {$a \in A$} $n_a \gets \min(n_a,n_c)$
  \EndFor
  \EndFor
\end{algorithmic}
\end{algorithm}

The following lemma proves the correctness of Algorithm~\ref{alg:xcoord}. Its running time is in $\Oh(\abs{A} n)$ since there are $n$ iterations of the main loop, and each iteration updates $\abs{A}$ counters.

\begin{lemma}\label{lem:xcorrect}
  Algorithm~\ref{alg:xcoord} computes the correct $x$-coordinates of the attributes of $a_1 \cdots a_n$.
\end{lemma}

\begin{proof}
  The algorithm reads the input word from left to right, letter by letter. In each step it updates some of its counters $n_a$.
  The semantics of the counters $n_a$ is as follows: if the next letter $a_i$ is $c$, then $x_i$ is $n_c$. This invariant is true after the initialization in the first line of the algorithm.
  
  Suppose that we start an iteration of the loop at letter $a_i = c$. Then the invariant tells us that $x_i = n_c$. If $a_{i+1}$ were $c$, then one more step $\X_{c}$ would be needed for a ranker to reach position $i+1$, hence $n_c \gets  n_c + 1$. If $a_{i+1}$ were some letter $a \neq c$, then we could either use the ranker corresponding to the old value $n_a$ or we could use the ranker going to position~$i$ and from there do an $\X_{a}$-modality; the latter would yield a ranker whose length is the new value of $n_c$. We choose the shorter of these two options. Since all counter values were correct before reading position $i$, there is no other counter $n_a$ which needs to be updated before proceeding with position $i+1$.
\end{proof}

With Algorithm~\ref{alg:ycoord}, we give 
a procedure for computing the
$y$-coordinates of $a_1 \cdots a_n$ similar to Algorithm~\ref{alg:xcoord}, but with the modification that we mark
some letters for deletion.
The positions marked for deletion depend on the number $k$ in Simon's
congruence $\sim_k$. The computed $y$-coordinates are those where all marked
letters are actually deleted. We assume that the $x$-coordinates of the input
word are already known.

\begin{algorithm}[h]
\caption{Computing the $y$-coordinates of $a_1 \cdots a_n$ plus deletion}\label{alg:ycoord}
\begin{algorithmic}[1]
  \ForAll {$a \in A$} $n_a \gets 1$
  \EndFor
  \For {$i \gets n,\ldots,1$} 
  \State suppose $a_i = c$
  \If {$x_i + n_c \leq k+1$} 
  \State $y_i \gets n_c$
  \State $n_c \gets  n_c + 1$
  \ForAll {$a \in A$} $n_a \gets \min(n_a,n_c)$
  \EndFor
  \Else
  \State position $i$ is marked for deletion
  \EndIf
  \EndFor
\end{algorithmic}
\end{algorithm}

The algorithm correctly computes the $y$-coordinates of the word where all marked letters are deleted. This follows from the left-right dual of Lemma~\ref{lem:xcorrect} and the fact that the counters remain unchanged if a position is marked for deletion.

\begin{lemma}
  Let $u$ be the input for Algorithm~\ref{alg:ycoord} and let $v$ be the word with all marked letters removed. Then $u \sim_k v$.
\end{lemma}

\begin{proof}
  Whenever a position $i$ with label $c$ is marked for deletion, the value $x_i$ is correct since no letter to the left of position $i$ is marked for deletion. The counter $n_c$ would be the correct $y$-coordinate for position $i$ if we deleted all positions which have been marked so far. By Proposition~\ref{prp:delLetter} we know that each deletion preserves the $\sim_k$-class.
\end{proof}

It remains to show that the $x$-coordinates are still correct for the resulting word in which all marked letters are deleted.

\begin{lemma}
  Consider a word $u = a_1 \cdots a_n$ with $x$-coordinate $x_\ell$ at position
  $\ell$. Let $i$ be the maximal position of $u$ such that $x_i + y_i > k + 1$
  and let $v = a_1 \cdots a_{i-1} a_{i+1} \cdots a_n$. The $x$-coordinate of
  position $\ell$ of $v$ is denoted by $x_\ell'$.
  Then, for all $j \in \os{i+1, \dots, n}$, we have $x_{j-1}' = x_j$.
\end{lemma}
\begin{proof}
  It suffices to prove the statement $x_{j-1}' = x_j$ for all positions $j$ of
  $u$ reachable by a ranker of the form $r \X_c$ with $r(u) = i$ and $c \in A$.
  By contradiction, suppose that there exists some position $j = r \X_c(u)$
  with $x_{j-1}' \ne x_j$ where $r(u) = i$ and $c \in A$; we choose $c \in A$
  such that $j$ is minimal with this property.
  Let $b = a_i$. We have to distinguish two cases.

  First suppose that there is no $b$-position $f$ with $i < f < j$ in $u$.
  The ranker $r^u_j$ has to visit position $i$ in $u$; otherwise $r^u_j(v) =
  j-1$ and $r^v_{j-1}(u) = j$, a contradiction to $x_{j-1}' \ne x_j$.
  This implies $x_i < x_j$.
  Moreover, position $i$ is reachable from $j$ in $u$ with a single
  $\Y_b$-modality, and hence, we have $x_i + y_i \leq (x_j - 1) + (y_j + 1)
  \leq k+1$. This contradicts the choice of $i$.

  Next, let $f$ be the minimal $b$-position with $i<f<j$. In particular, we
  have $b \neq c$ because $j \ne f$ is the smallest $c$-position of $u$ greater
  than $i$.
  Let $r \X_c$ be an \Xra of length $x_j$ such that $r \X_c (u) = j$.
  If $r(u) = i$, then $r(v) = f - 1$ and hence $r \X_c (v) = j - 1$.
  If $r(u) < i$, then the ranker $r \X_c$ does not visit the position $i$ in
  $u$ and we have $r \X_c(v) = j - 1$.
  Finally, if $r(u) > i$, then (by choice of $c$) the position $r(u) < j$ keeps
  its $x$-coordinate. In other words, there exists an \Xra $r'$ with
  $\abs{r}=\abs{r'}$ and $r'(v) = r(u) - 1$.
  It follows that $r' \X_c(v) = j - 1$.  Therefore, in any case, there exists a
  ranker $s$ of length at most $x_j$ such that $s(v) = j - 1$. This shows
  $x'_{j-1} \leq x_j$, and together with $x'_{j-1} \neq x_j$ we obtain
  $x'_{j-1} < x_j$.

  Consider an \Xra $s \X_c$ of length $x'_{j-1} < x_j$ with $s \X_c(v) = j - 1$.
  We are still in the situation that there exists a $b$-position $f$ in $u$
  with $i<f<j$.
  We cannot have $s(v) < i$ since otherwise $s(u) = s(v)$ and, thus, $s
  \X_c(u) = j$; the latter uses the fact that $b \neq c$.
  Let now $s(v) \ge i$ and write $s = t \X_d$. We have $t(v) < i$ since
  otherwise $t \X_c$ would be a shorter \Xra with $t \X_c(v) = j - 1$.
  We have $d = b$: if $d \neq b$, then $s(v) = s(u) - 1$ and $s \X_c(u) = j$;
  this would show $x'_{j-1} \geq x_j$, thereby contradicting $x'_{j-1} < x_j$.
  It follows that $s(u) = i$ and $s \X_c(u) = j$. As before, this is a
  contradiction. This completes the proof that $x'_{j-1} = x_j$.
\end{proof}

\begin{example}
Let $u = bacbaabada$ be the word from Example~\ref{ex:second} and let $k = 3$.
Suppose that the alphabet $A = \os{a, b, c, d}$ is ordered by $a < b < c < d$.
The attributes of $u$ are as follows:
\begin{equation*}
  \aletter{b}{1}{2}
  \aletter{a}{1}{2}
  \aletter{c}{1}{1}
  \aletter{b}{2}{2}
  \aletter{a}{2}{3}
  \aletter{a}{3}{2}
  \aletter{b}{3}{1}
  \aletter{a}{4}{2}
  \aletter{d}{1}{1}
  \aletter{a}{2}{1}
\end{equation*}
Note that each of the attributes $(x_i, y_i)$ at positions $i \in \os{5, 6, 8}$
satisfies the condition $x_i + y_i > k + 1$. As seen in
Example~\ref{ex:double-del} we must not delete all these positions.
The algorithm only marks positions $6$ and $8$ for deletion and takes these
deletions into account when computing the $y$-coordinates of the remaining
letters:
\begin{equation*}
  \aletter{b}{1}{2}
  \aletter{a}{1}{2}
  \aletter{c}{1}{1}
  \aletter{b}{2}{2}
  \aletter{a}{2}{2}
  \maletter{a}{3}{\phantom{3}}
  \aletter{b}{3}{1}
  \maletter{a}{4}{\phantom{4}}
  \aletter{d}{1}{1}
  \aletter{a}{2}{1}
\end{equation*}
The letters are now sorted as in Example~\ref{ex:commutation} and the resulting
normal form is $bacabbda$.
\qed
\end{example}

The following lemma allows us to improve the estimated time for the sorting
step of the main algorithm by showing that any sequence of letters which needs to be sorted contains every letter at most once.

\begin{lemma}
  Consider a word $uaav$ with $a \in A$ and $\abs{ua} = i$. Then $x_i \neq x_{i+1}$ and $y_i \neq y_{i+1}$.
  \label{lem:no-aa}
\end{lemma}

\begin{proof}
Suppose $x_i = x_{i+1}$. Let $r_i$ and $r_{i+1}$ be \Xras with $r_i(uaav) = i$, $\abs{r_i} = x_i$, $r_{i+1}(uaav) = i+1$, and $\abs{r_{i+1}} = x_{i+1} = x_i$. Let $r_{i+1} = s \X_a$. If $s(uaav) < i$, then $i+1 = r_{i+1}(uaav) = s \X_a (uaav) \leq i$. 
If $s(uaav) = i$, then $\abs{s} = x_i - 1 < x_i = \abs{r_i}$ contradicts the definition of $x_i$. Therefore, we cannot have $x_i = x_{i+1}$. Symmetrically, we cannot have $y_i = y_{i+1}$.
\end{proof}

We are now able to state our main result.

\begin{theorem}
  One can compute the shortlex normal form of a word $w$ of length $n$,
  including all attributes of the normal form, with $\Oh(\abs{A} n)$
  arithmetic operations and with bit complexity $\Oh(\abs{A} n \log n)$.
  Alternatively, the computation can be done in deterministic space
  $\Oh(\abs{A} \log n)$.
  \label{thm:main}
\end{theorem}

\begin{proof}
  The attributes of the normal form can be computed as described in
  Algorithms~\ref{alg:xcoord} and~\ref{alg:ycoord}. The normal form itself is
  obtained by filtering out all positions $i$ where the corresponding attribute
  $(x_i, y_i)$ satisfies $x_i + y_i \le k + 1$ and by sorting blocks of letters
  with the same attributes satisfying $x_i + y_i = k + 1$. By
  Lemma~\ref{lem:no-aa}, the sorting step can be performed by reading each such
  block of letters, storing all letters appearing in the block and only
  outputting all these letters in sorted order once the next block is reached.

  If we assume that the comparison of two letters and the modification of the counters is possible in constant time, then running Algorithm~\ref{alg:ycoord} on the output of Algorithm~\ref{alg:xcoord} takes $\Oh(\abs{A} n)$ steps for input words of length $n$ over alphabet~$A$: for each position of the input word, we need to update $\abs{A}$ counters. Over fixed alphabet, the resulting algorithm runs in linear time\,---\,even if $k$ is part of the input. We could bound all arithmetic operations by $k+2$, \ie, by replacing the usual addition by $n \oplus m = \min(k+2,n+m)$. This way, each counter and all results of arithmetic operations would require only $\Oh(\log k) \subseteq \Oh(\log n)$ bits. Similarly, $\Oh(\log \abs{A}) \subseteq \Oh(\log n)$ bits are sufficient to encode the letters. This leads to a bit complexity of $\Oh(\abs{A} n \log n)$. 
Note that if $k > n$, then the $\sim_k$-class of the input is a singleton and we can immediately output the input without any further computations. If $\abs{A} > n$, then we could replace $A$ by the letters which occur in the input word.

  For the $\Oh(\abs{A} \log n)$ space algorithm, one can again use
  Algorithms~\ref{alg:xcoord} and~\ref{alg:ycoord} to compute the attributes of
  each position.
  To compute the shortlex normal form, we do not store all the attributes but
  use the standard recomputation technique to decide whether a letter gets
  deleted.
  The sorting step can be implemented by repeatedly scanning each block of
  positions with common attributes $(x, y)$ satisfying $x + y = k + 1$. A
  single scan checks, for a fixed letter $a \in A$, whether $a$ occurs in the
  block. This is repeated for every $a \in A$ in ascending order. The
  attributes of the currently investigated block and the current letter $a$ can
  be stored in space $\Oh(\log n)$.
\end{proof}

\section{Computing minimal length rankers}\label{sec:rankercomputation}

Canonical rankers are the crucial ingredient in proving the completeness parts of both the deletion procedure and the commutation principle. One aim of this section is to show that canonical rankers are actually a very natural concept. We illustrate this claim by showing how to derive the computation of the natural rankers from the computation of all length-minimal rankers reaching certain positions.

Algorithm~\ref{alg:xcoord} from Section~\ref{sec:alg} can be modified such that
it additionally computes the sets of rankers $R^u_i$ for all positions $i$ of
an input word $u$; remember that $R^u_i$ is the set of all \Xras with  minimal length
reaching position $i$. Of course, by left-right symmetry, this computation can
also be adapted for computing all \Yras of minimal length.

For every position $i$, in addition to the first component $x_i$ of its
attribute, we also compute a set of positions $P_i$ containing all immediate
predecessors of the minimal length rankers reaching $i$. Note that, if $i$ is a
$c$-position, then the last modality of all rankers reaching $i$ is $\X_c$; in
particular, we do not need to store the modalities. In addition to the counters~$n_a$ for~$a \in A$, the algorithm also uses sets of positions $Q_a$. We say
that a position~$j$ of $u = a_1 \cdots a_n$ is a predecessor of $i$ if there
exists a minimal length \Xra $r \ms \X_c$ reaching~$i$ such that $r(u) = j$.

\begin{algorithm}[h]
\caption{Computing all minimal length \Xras of $a_1 \cdots a_n$}\label{alg:xcoordminlen}
\begin{algorithmic}[1]
  \ForAll {$a \in A$} $n_a \gets 1$; \,$Q_a \gets \emptyset$
  \EndFor
  \For {$i \gets 1,\ldots,n$} 
  \State suppose $a_i = c$
  \State $x_i \gets n_c$; \,$P_i \gets Q_c$
  \State $n_c \gets  n_c + 1$; \,$Q_c \gets \os{i}$
  \ForAll {$a \in A$} 
  \If {$n_c < n_a$} $n_a \gets n_c$; \,$Q_a \gets Q_c$
  \ElsIf {$n_c = n_a$} $Q_a \gets Q_a \cup Q_c$
  \EndIf
  \EndFor
  \EndFor
\end{algorithmic}
\end{algorithm}

To simplify the notation in the following proposition, we assume that $\bigcup_{j \in \emptyset} R^u_j \ms \X_c = \os{\X_c}$.

\begin{proposition}\label{prp:Rui}
  Let $u$ be the input for Algorithm~\ref{alg:xcoordminlen} and let $P_i$ be the set of positions computed for position $i$.
  We have $R^u_i = \bigcup_{j \in P_i} R^u_j \ms \X_c$ if $i$ is a $c$-position.
\end{proposition}

\begin{proof}
  We can assume that the semantics of the counters $n_c$ is the one given in the proof of Lemma~\ref{lem:xcorrect}.
  It suffices to show that $P_i$ contains all predecessors of $i$. For this purpose, we prove the following invariant. Before entering the iteration of the loop with position $i$, if~$a_i = c$, then $Q_c$ contains all predecessors of $i$. This invariant is true after the initialization.
  
  Suppose that $a_i = c$ and that the invariant is true before entering the $i$-th iteration. We have to show that it is also true after the $i$-th iteration (that is, before entering iteration~$i+1$). If $a_{i+1} = c$, every $\X_c$-modality reaching $i+1$ cannot start at a position smaller than~$i$; hence, in this case, position $i$ is the unique predecessor of $i+1$.
  
  Note that the inner for loop does not change $Q_c$.
  Let $a \neq c$, suppose $a_{i+1} = a$ and consider the counter $n_a$ before entering the $i$-th iteration.
  If $n_a < x_i + 1$, then position $i$ is not a predecessor of $i+1$ and the predecessors of $i+1$ are the same as in the case where~$i$ would have been an $a$-position.
  If $n_a > x_i + 1$, then position $i$ is the only predecessor of~$i+1$ since all other positions would give rankers of length at least $n_a$ (which is not minimal in this case). 
  If $n_a = x_i + 1$, then $Q_a \cup \os{i}$ contains all predecessors of $i+1$. In all cases, the program variable $Q_a$ is updated to the correct set of predecessors.
\end{proof}

By recursively applying Proposition~\ref{prp:Rui} we get a presentation for $R^u_i$ for each position $i$. The following lemma shows that each set of predecessors is bounded by the alphabet. In particular, this bound also applies to the size of the variables~$Q_a$.

\begin{lemma}
  Let $P_i$ be the set of all predecessors of a position $i$. Then the positions in $P_i$ all have different labels.
\end{lemma}

\begin{proof}
If there are $a$-positions $k<\ell$ which are possible predecessors of a $c$-position $i$, then this leads to a shortcut to $i$ which is a contradiction: The positions $k$ and $\ell$ have the same $x$-coordinate, say $m$. In particular, the $x$-coordinate of $i$ is $m+1$. In order to reach $\ell$ in $m$ steps, there has to exist a position $j$ with $k<j<\ell$ with an $x$-coordinate smaller than $m$. Going to $j$ and then (with only one $\X_c$-modality) to $i$ yields a ranker of length at most $m$ reaching position $i$, in contradiction to its $x$-coordinate $m+1$.
\end{proof}

If instead of the sets $Q_a$ in Algorithm~\ref{alg:xcoordminlen}, we only keep
their minimal positions~$q_a$, then this exactly computes the canonical \Xras
$r^u_i$, see Algorithm~\ref{alg:xcoordcanonical}. While the definition of the
canonical rankers minimizes from right to left, the algorithm 
processes the positions from left to right. The latter direction is coherent
with the statement in Lemma~\ref{lem:consecutive}.

\begin{algorithm}[h]
\caption{Computing the $x$-coordinates and canonical rankers of $a_1 \cdots a_n$}\label{alg:xcoordcanonical}
\begin{algorithmic}[1]
  \ForAll {$a \in A$} $n_a \gets 1$; \,$q_a \gets 0$
  \EndFor
  \For {$i \gets 1,\ldots,n$} 
  \State suppose $a_i = c$
  \State $x_i \gets n_c$; \,$p_i \gets q_c$
  \State $n_c \gets  n_c + 1$; \,$q_c \gets i$
  \ForAll {$a \in A$} 
  \If {$n_c < n_a$} $n_a \gets n_c$; \,$q_a \gets q_c$
  \EndIf
  \EndFor
  \EndFor
\end{algorithmic}
\end{algorithm}

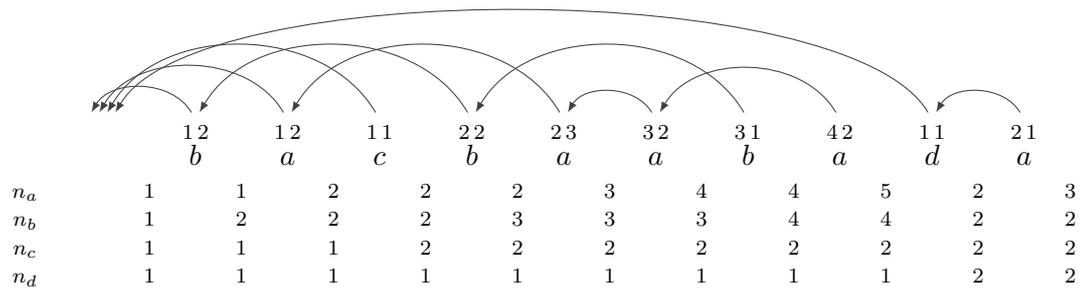
\begin{figure}[ht]
  \centering
  \begin{tikzpicture}
    \matrix (A) [matrix of math nodes,column sep=0mm,row sep=0mm]{%
      \hspace{1cm}
      & \phantom{\aletter{b}{1}{1}} &
      & \aletter{b}{1}{2} &
      & \aletter{a}{1}{2} &
      & \aletter{c}{1}{1} &
      & \aletter{b}{2}{2} &
      & \aletter{a}{2}{3} &
      & \aletter{a}{3}{2} &
      & \aletter{b}{3}{1} &
      & \aletter{a}{4}{2} &
      & \aletter{d}{1}{1} &
      & \aletter{a}{2}{1} \\
      {\scriptstyle n_a}
      & &
      {\scriptstyle 1} & &
      {\scriptstyle 1} & &
      {\scriptstyle 2} & &
      {\scriptstyle 2} & &
      {\scriptstyle 2} & &
      {\scriptstyle 3} & &
      {\scriptstyle 4} & &
      {\scriptstyle 4} & &
      {\scriptstyle 5} & &
      {\scriptstyle 2} & &
      {\scriptstyle 3} \\[-1mm]
      {\scriptstyle n_b}
      & &
      {\scriptstyle 1} & &
      {\scriptstyle 2} & &
      {\scriptstyle 2} & &
      {\scriptstyle 2} & &
      {\scriptstyle 3} & &
      {\scriptstyle 3} & &
      {\scriptstyle 3} & &
      {\scriptstyle 4} & &
      {\scriptstyle 4} & &
      {\scriptstyle 2} & &
      {\scriptstyle 2} \\[-1mm]
      {\scriptstyle n_c}
      & &
      {\scriptstyle 1} & &
      {\scriptstyle 1} & &
      {\scriptstyle 1} & &
      {\scriptstyle 2} & &
      {\scriptstyle 2} & &
      {\scriptstyle 2} & &
      {\scriptstyle 2} & &
      {\scriptstyle 2} & &
      {\scriptstyle 2} & &
      {\scriptstyle 2} & &
      {\scriptstyle 2} \\[-1mm]
      {\scriptstyle n_d}
      & &
      {\scriptstyle 1} & &
      {\scriptstyle 1} & &
      {\scriptstyle 1} & &
      {\scriptstyle 1} & &
      {\scriptstyle 1} & &
      {\scriptstyle 1} & &
      {\scriptstyle 1} & &
      {\scriptstyle 1} & &
      {\scriptstyle 1} & &
      {\scriptstyle 2} & &
      {\scriptstyle 2} \\
    };

    \draw[color=darkgray]
      ([xshift=-1.5pt]A-1-4.north) edge[-latex,out=120,looseness=1,in=60] ([xshift=-4.5pt]A-1-2.north)
      ([xshift=-1.5pt]A-1-6.north) edge[-latex,out= 120,looseness=1,in=60] ([xshift=-1.5pt]A-1-2.north)
      ([xshift=-1.5pt]A-1-8.north) edge[-latex,out= 120,looseness=1,in=60] ([xshift=1.5pt]A-1-2.north)
      ([xshift=-1.5pt]A-1-10.north) edge[-latex,out=120,looseness=1,in=60] ([xshift=1.5pt]A-1-4.north)
      ([xshift=-1.5pt]A-1-12.north) edge[-latex,out=120,looseness=1,in=60] ([xshift=1.5pt]A-1-6.north)
      ([xshift=-1.5pt]A-1-14.north) edge[-latex,out=120,looseness=1,in=60] ([xshift=1.5pt]A-1-12.north)
      ([xshift=-1.5pt]A-1-16.north) edge[-latex,out=120,looseness=1,in=60] ([xshift=1.5pt]A-1-10.north)
      ([xshift=-1.5pt]A-1-18.north) edge[-latex,out=120,looseness=1,in=60] ([xshift=1.5pt]A-1-14.north)
      ([xshift=-1.5pt]A-1-20.north) edge[-latex,out=120,looseness=.5,in=60] ([xshift=4.5pt]A-1-2.north)
      ([xshift=-1.5pt]A-1-22.north) edge[-latex,out=120,looseness=1,in=60] ([xshift=1.5pt]A-1-20.north)
      ;
  \end{tikzpicture}
  \caption{Computation of the $x$-coordinates of $bacbaabada$}
  \label{fig:xcoord}
\end{figure}

The computation of Algorithm~\ref{alg:xcoordcanonical} on input $bacbaabada$, including the values of
each counter $n_a$ for $a \in A$ and arrows representing the pointers $p_i$, is depicted in Figure~\ref{fig:xcoord}.

\section{Summary and Outlook}

We considered Simon's congruence $\sim_k$ for piecewise testable languages.
The main contribution of this paper is an $\Oh(\abs{A}n)$ algorithm for computing the shortlex normal form of a word of length $n$ within its $\sim_k$-class; surprisingly, this bound also holds if $k$ is part of the input. The algorithm can be adapted to work in deterministic logarithmic space over fixed alphabet. As a consequence, on input $u,v,k$, one can test in time $\Oh(\abs{A}\ms\abs{uv})$ whether $u \sim_k v$ holds. The main tool are the minimal lengths of \Xras and \Yras reaching any position of a word. The key ingredient in the proofs are the so-called canonical rankers.
In Section~\ref{sec:rankercomputation}, we give some additional insight into this concept by providing an algorithm for their computation.

It would be
interesting to see whether the space complexity for an arbitrary alphabet can
be further improved from $\Oh(\abs{A} \log n)$ to nondeterministic log-space
or even deterministic log-space if the alphabet $A$ is part of the input.
In addition, we still lack corresponding lower bounds for the computation of shortlex normal forms and for the test of whether $u \sim_k v$ holds.

\newcommand{\Ju}{Ju}\newcommand{\Ph}{Ph}\newcommand{\Th}{Th}\newcommand{\Ch}{Ch}\newcommand{\Yu}{Yu}\newcommand{\Zh}{Zh}\newcommand{\St}{St}\newcommand{\curlybraces}[1]{\{#1\}}

\end{document}